\documentclass[conference,a4paper]{IEEEtran}
\usepackage{amsfonts,amssymb} 
\usepackage{amsmath}
\usepackage{graphicx}
\usepackage{comment}
\usepackage{url}

\begin{document}

\newtheorem{theorem}{Theorem}
\newtheorem{corollary}{Corollary}
\newtheorem{proposition}{Proposition}
\newtheorem{lemma}[theorem]{Lemma}
\newtheorem{assumption}{Assumption}


\def\IN{ {\Lambda_N}^* }
\def\muN{ {\mu_N} }
\def\IGN{ {\Lambda_{NG}^* } }

\def\HG{ {H_G} }

\def\LambdaG{ {\Lambda_G} }
\def\LambdaN{ {\Lambda_N} }
\def\LambdaGN{ {\Lambda_{NG}} }

\def\A{ {\mathbb{A} }}
\def\N{ {\mathbb{N} }}
\def\R{ {\mathbb{R} }}

\def\TURN{ { \gamma } }
\def\E{{E}}
\def\P{{P}}
\def\T{{T_{\epsilon,k}}}
\def\T{{T^\epsilon_k}}
\def\U{U^\epsilon}
\def\W{W^\epsilon}
\def\p1{ {g} }
\def\lplus{l^+}
\def\lminus{l^-}
\def\norm{ \eta }

\def\SIZE{ {N_k(l) } }
\def\SIZEp{ {N_k(\lplus) } }
\def\SIZEm{ {N_k(\lminus) } }
\def\SIZEt{ {N_k(l^*) } }
\def\SIZEemax{ {N_k(l^{\epsilon, *, (k)}} }
\def\SIZEmax{ {N_k(l^{*, (k)}) } }

\def\Na{ {n_k(w,a)} }
\def\La{ {L_k} }
\def\Le{ {L_{\epsilon, k}} }
\def\Les{ {L_{\epsilon}} }

\def\lemax{l^{\epsilon, *, (k)}}
\def\lmax{l^{(k)}}

\sloppy

\title{Guessing a password over a wireless channel (on the effect
of noise non-uniformity)}

\author{
    \IEEEauthorblockN{Mark M. Christiansen and Ken R. Duffy}
   \IEEEauthorblockA{Hamilton Institute\\
     National University of Ireland, Maynooth\\
     Email: \{mark.christiansen, ken.duffy\}@nuim.ie}\\
  \and
\IEEEauthorblockN{Fl\'avio du Pin Calmon and Muriel M\'edard }
  \IEEEauthorblockA{Research Laboratory of Electronics\\
Massachusetts Institute of Technology\\
     Email: \{flavio, medard\}@mit.edu}
  
}


\maketitle

\begin{abstract}
A string is sent over a noisy channel that erases some of its
characters. Knowing the statistical properties of the string's
source and which characters were erased, a listener that is equipped
with an ability to test the veracity of a string, one string at a time,
wishes to fill in the missing pieces. Here we characterize the
influence of the stochastic properties of both the string's source
and the noise on the channel on the distribution of the number of
attempts required to identify the string, its guesswork. In particular,
we establish that the average noise on the channel is not a determining
factor for the average guesswork and illustrate simple settings
where one recipient with, on average, a better channel than another
recipient, has higher average guesswork. These results stand in
contrast to those for the capacity of wiretap channels and suggest
the use of techniques such as friendly jamming with pseudo-random
sequences to exploit this guesswork behavior.

\end{abstract}

\section{Introduction}
\let\thefootnote\relax\footnote{
F.d.P.C. and M.M. sponsored by the Department of Defense under Air
Force Contract FA8721-05-C-0002. Opinions, interpretations,
recommendations, and conclusions are those of the authors and are
not necessarily endorsed by the United States Government. Specifically,
this work was supported by Information Systems of ASD(R\&E). }

This paper quantifies the influence of the stochastic properties
of a string source and of a noisy erasure channel on the difficulty
a listener has in guessing the erased pieces of the string. As a
concrete example in advance of the mathematical abstraction, consider
a proximity card reader where an electronic signature, a password,
is wirelessly transmitted when the card is near the reader. An
unintended recipient is eavesdropping, but overhears the card's
transmission via a noisy channel that erases certain characters.
If the eavesdropper knows the string's source statistics and which
characters were erased, how many guesses must he make before
identifying the one that causes the card reader to notify success?

For i.i.d. character sources and noise that is independent of the
string, but possibly correlated, Theorem \ref{thm:subguess} answers
this question, providing an asymptotic approximation to the guesswork
distribution as the string becomes long. Corollary \ref{cor:growth}
establishes that the mean number of erasures on the channel and the
Shannon entropy of the character source determine the growth rate
of the expected logarithm of the number of guesses required to
identify the erased sub-string. The exponential growth rate of the
average number of guesses, however, is determined by the scaling of
the asymptotic moment of the number of erasures evaluated at the
R\'enyi entropy, with parameter $1/2$, of the character distribution.

As a consequence of these results, we provide examples illustrating
that the average guesswork can be smaller on a channel that is, on
average, noisier demonstrating that average noise is not a useful
statistic for guesswork. This conclusion may seem counterintuitive
in the context of capacity results for Wyner's wire-tap \cite{Wyner75}
that, when applied to an erasure channel, indicate that secrecy
capacity is non-zero only if the probability of erasure of the
intended party is lower than that of the eavesdropper. Results in
which a first receiver, with more erasures (on average) than a
second receiver, can better recover a message than the second
receiver are, to the authors' knowledge, few. One recent exception
is \cite{Czap11}, which also considers the effect of randomness of
erasures in message recovery. In contrast to our work, the authors
consider secret message capacity in a specific setting that uses
feedback to provide causal channel state information for the intended
receiver, allowing the sender to transmit in a way that is advantageous
to the intended receiver. In the case of two parties with an erasure,
their scheme relies on the fact that the secret key agreement by
public discussion from common information developed by \cite{Maurer93}
reduces to requiring only the channel state be shared over a public
channel.

Guesswork analysis of a distinct wiretap model to the one considered
here is provided in \cite{Merhav99,Hanawal11a}. There it is assumed
that an eavesdropper observes a string that has been encrypted with
a function employing a uniformly chosen random key. The impact of
key rate on asymptotic moments of guessing is determined.

\section{Guesswork and erasure channels}

We begin with summarizing material on the mathematical formulation
for guesswork followed by a brief overview of the relevance of
erasure channels as models of wireless communication.
Let $\A=\{0,\ldots,m-1\}$ be a finite alphabet and consider a
stochastic sequence of words, $\{W_k\}$, where $W_k$ is a string
of length $k$ taking values in $\A^k$. Assume that a word is selected
and an inquisitor is equipped with a device, such as a one-way hash
function, through which a word can be tested one at a time. With
no information beyond the string length and the source statistics,
their optimal strategy to identify the word is to generate a
partial-order of the words from most likely to least likely and
guess them in turn. That is, for each $k$ the attacker generates a
function $G:\A^k\to\{1,\ldots,m^k\}$ such that $G(w')<G(w)$ if
$\P(W_k=w')>\P(W_k=w)$. For a word $w$ the integer $G(w)$ is the
number of guesses until the string $w$ is guessed, its Guesswork.

Massey \cite{Masey94} established that the Shannon entropy of the
word source bears little relation to the average guesswork.  As
word length grows an asymptotic relationship between scaled moments
of the guesswork distribution and specific R\'enyi entropy was
identified under weak assumptions on the stochastic properties of
the string source \cite{Arikan96,Malone042,Pfister04,Hanawal11}.
These results have recently been built upon to establish that
$\{k^{-1}\log G(W_k)\}$ satisfies a Large Deviation Principle (LDP)
\cite{Christiansen13}, giving a direct approximation to the guesswork
distribution, $P(G(W_k)=n)$ for $n\in\{1,\ldots,m^k\}$.

In the present article we restrict to i.i.d. letter sources, but
include noise sources that could potentially be correlated. This
enables us to consider the erasures as a subordinating process for
the guesswork, as will become clear.
\begin{assumption}
The string $W_k$ is constituted of independent and identically
distributed characters with distribution $P(W_1=i)$ for $i\in\A$.
\end{assumption}
Under this assumption, if one must guess the entire word $W_k$, 
the following result is known.
\begin{proposition}[\cite{Arikan96,Pfister04,Christiansen13}]
The scaled Cumulant Generating Function (sCGF) of $\{k^{-1}\log
G(W_k)\}$ exists 
\begin{align}
\LambdaG(\alpha)&=
\lim_{k\to\infty} \frac 1k \log E(\exp(\alpha\log(G(W_k)))) \nonumber\\
	&= 
	\begin{cases}
	\displaystyle
	\alpha R\left(\frac{1}{1+\alpha}\right) & \text{ if } \alpha>-1\\
	-R(\infty) & \text{ if } \alpha\leq-1,\\
	\end{cases}
\label{eq:sCGF}
\end{align}
where $R(\alpha)$ is the R\'enyi entropy with parameter $\alpha$,
\begin{align*}
R(\alpha) &= \frac{1}{1-\alpha}\log\left(\sum_{i\in\A} P(W_1=i)^\alpha\right)\\
R(\infty) &= -\max_{i\in\A}\log P(W_1=i).
\end{align*}
Moreover, the process $\{k^{-1}\log G(W_k)\}$
satisfies a Large Deviation Principle with rate function
\begin{align}
\label{eq:rf}
\LambdaG^*(x) = \sup_{\alpha\in\R}(x\alpha-\LambdaG(\alpha)).
\end{align}
\end{proposition}

As in \cite{Arikan96}, setting $\alpha=1$ equation
\eqref{eq:sCGF} gives
\begin{align*}
\LambdaG(1) &= \lim_{k\to\infty} \frac 1k \log E(G(W_k)) \\
	&= R(1/2) 
	= 2\log\left(\sum_{i\in\A} P(W_1=i)^{1/2}\right),
\end{align*}
establishing that the exponential growth rate of the average guesswork
as the string gets longer is governed by R\'enyi entropy of the
character distribution with parameter $1/2$, which is greater than
its Shannon entropy, with equality if and only if the character
source is uniformly distributed. The LDP gives the following
approximation \cite{Christiansen13} for large $k$ and
$n\in\{1,\ldots,m^k\}$,
\begin{align*}
P(G(W_k)=n) \approx \frac 1n \exp\left(-k\LambdaG^*\left(\frac 1k\log n\right)\right). 
\end{align*}

Erasure models are common for coded communications. They arise for
systems where an underlying error-correcting code can fail to correct
the errors, but error-detection mechanisms will lead to detection
of the failure to correct. While it is possible for errors to remain
uncorrected in such a way that the receiver cannot detect the failure
to correct. That traditional algebraic codes with $n$ symbols of
redundancy can correct up to $n$ errors but detect up to $2n-1$
errors justifies the common assumption that failures to detect
errors may be neglected, whereas failures to correct may not. Failure
to correct errors may be a design goal in certain systems. In wiretap
channels, codes are deliberately constructed in such a way that,
under channel conditions less favorable than those of the intended
receiver, codes fail to decode (e.g. \cite{Bloch11}).

\section{Subordinated guesswork - General Results}

We wish to consider the guesswork required to identify a string,
$W_k$, sent over a stochastic, noisy channel that erases characters.
We assume that a listener is equipped with an ability to test the
veracity of each missing sub-string and wishes to fill in the missing
piece. As the string $W_k$ is made up of i.i.d. characters, if
$N_k\in\{1,\ldots,k\}$ is the number of characters erased by the
noise, the listener must effectively guess a word of $N_k$ characters
in length. Thus we are interested in properties of the the guesswork
of the word subordinated by the erasures process, $G(W_{N_k})$, 
wishing to understand the influence of the properties of the string
source and the noise on the channel on the distribution of the
number of attempts required to identify the missing sub-string.

While we assume that the string is made up of i.i.d. characters,
the noise process can be correlated and we make the following
assumption, which encompasses, for example, Markovian erasure
processes. 
\begin{assumption}
\label{ass:noise}
The noise process is such that $\{N_k/k\}$, where $N_k$ is
the number of erasures due to noise in a string of length $k$,
satisfies a LDP with convex rate function $\IN:\R\mapsto[0,\infty]$
such that $\IN(y)=\infty$ if $y\notin[0,1]$.
\end{assumption}
Loosely speaking, assumption \ref{ass:noise} implies that $P(N_k\approx
yk) \asymp \exp(-k \IN(y))$. 
Our main general, structural result is the following. 
\begin{theorem}
\label{thm:subguess}
The subordinated guesswork process $\{1/k \log G(W_{N_k})\}$ 
satisfies a LDP with convex rate function
\begin{align}
\label{eq:IGN}
\IGN(x) = \inf_{y\in[0,1]}
	\left(y\LambdaG^*\left(\frac{x}{y}\right)+\IN(y)\right).
\end{align}
The sCGF for $\{1/k \log G(W_{N_k})\}$, the Legendre-Fenchel
transform of $\IGN$, is given by the composition of the
sCGF for the noise with sCGF for the non-subordinated guesswork
\begin{align}
\label{eq:LambdaGN}
\LambdaGN(\alpha) 
&=
\lim_{k\to\infty} \frac 1k \log E\left(\exp(\alpha\log(G(W_{N_k})))\right)
	\nonumber\\
	&= \LambdaN(\LambdaG(\alpha)).
\end{align}
\end{theorem}
\begin{proof}
The method of proof of the LDP is akin to that used in
\cite{Christiansen13}, establishing that the upper and lower deviation
functions coincide, followed by an application of the contraction
principle. With $B_\epsilon(x)=(x-\epsilon,x+\epsilon)$. 
We first show that
\begin{align*}
&\lim_{\epsilon\downarrow0} 
	\liminf_{k\to\infty} 
	\frac 1k \log 
	P\left(\frac 1k \log G(W_{N_k}) \in B_\epsilon(x),
		\frac{N_k}{k}\in B_\epsilon(y) \right)\\
&=
\lim_{\epsilon\downarrow0} \limsup_{k\to\infty} 
	\frac 1k \log 
	P\left(\frac 1k \log G(W_{N_k}) \in B_\epsilon(x),
		\frac{N_k}{k}\in B_\epsilon(y) \right)\\
&= 
	y\LambdaG^*\left(\frac{x}{y}\right)+\IN(y)
	\text{ for all } x\geq0, y\in[0,1].
\end{align*} 
For example, for $y\in(0,1]$, consider
\begin{align*} 
&\frac 1k \log 
	P\left(\frac 1k \log G(W_{N_k}) \in B_\epsilon(x),
		\frac{N_k}{k}\in B_\epsilon(y) \right)\\
	\geq &
\frac 1k \log 
	P\left(\frac 1k \log G(W_{\lfloor k(y-\epsilon)\rfloor}) 
		\in B_\epsilon(x)\right)\\
&+
\frac 1k \log P\left(\frac{N_k}{k}\in B_\epsilon(y) \right).
\end{align*} 
Taking $\liminf_{k\to\infty}$, using the LDPs for $\{k^{-1}\log
G(W_k)\}$ and $\{N_k/k\}$ followed by $\lim_{\epsilon\downarrow0}$
gives an appropriate lower bound. An equivalent upper bound follows
similarly. 

For $y=0$, if $x>0$ we can readily show that the upper deviation
function takes the value $-\infty$ as $G(W_{\lfloor\epsilon
y\rfloor})\leq m^{y\epsilon}$. If $x=0$, then the $\limsup$ bound
is achieved by solely considering the noise term, while for the
$\liminf$ consider the ball $G(W_{N_k})\leq \exp(k\epsilon\log(m))$,
which has probability $1$ and so the upper and lower deviation
functions again coincide. 

As the state space is compact, the LDP for $\{(1/k \log G(W_{N_k}),
N_k/k)\}$ follows (e.g. \cite{Lewis95A, Dembo98}) with the rate
function $y\LambdaG^*(x/y)+\LambdaN^*(y)$. From this LDP, the LDP
for $\{(1/k \log G(W_{N_k})\}$ via the contraction principle
\cite{Dembo98} by projection onto the first co-ordinate.

To prove that $\IGN(x)$ is convex in $x$, first note that
$y\LambdaG^*(x/y)$ is jointly convex in $x$ and $y$, with $y>0$,
by the following argument. For $\beta\in(0,1)$,
set $\eta = \beta y_1/(\beta y_1+(1-\beta)y_2)\in[0,1]$ and note that
\begin{align*}
&(\beta y_1+(1-\beta)y_2)
\LambdaG^*\left(\frac{\beta x_1+(1-\beta)x_2}{\beta y_1+(1-\beta)y_2}\right)\\
&= (\beta y_1+(1-\beta)y_2) 
	\LambdaG^*\left(\eta \frac{x_1}{y_1}+(1-\eta){x_2}{y_2}\right)\\
&\leq \beta y_1 \LambdaG^*\left(\frac{x_1}{y_1}\right)
	+(1-\beta)y_2 \LambdaG^*\left(\frac{x_2}{y_2}\right),
\end{align*}
where we have used the convexity of $\LambdaG^*$. As the sum of
convex functions is convex, $y\LambdaG^*\left(x/y\right)+\IN(y)$ is
convex and as the point-wise minimum of a jointly convex function
is convex, $\IGN(x)$ is convex.

An application of Varadhan's Lemma (Theorem 4.3.1 \cite{Dembo98})
identifies the sCGF for the subordinated process as the Legendre
Fenchel transform of $\IGN$, $\sup_{x\in\R}(\alpha x - \IGN(x))$. 
To convert this into an expression in terms of $\LambdaN$ and
$\LambdaG$ observe that
\begin{align*}
\sup_{x\in\R}(\alpha x - \IGN(x)) 
&= 
\sup_{x\in\R}\sup_{y\in\R}
	\left(\alpha x- y\LambdaG^*\left(\frac{x}{y}\right)-\IN(y)\right)\\
&= 
\sup_{y\in\R}
	\left(y \sup_{z\in\R}(\alpha z- \LambdaG^*(z))-\IN(y)\right)\\
&= 
\sup_{y\in\R}
	\left(y \LambdaG(\alpha)-\IN(y)\right)\\
&= 
	\LambdaN(\LambdaG(\alpha)).
\end{align*}
\end{proof}
Theorem \ref{thm:subguess}, in particular, identifies the
growth rate of the average subordinated guesswork.
\begin{corollary}
\label{cor:growth}
The growth rate of the average of the logarithm of the
subordinated guesswork is determined by the average noise
and the Shannon entropy of the character source
\begin{align*}
\lim_{k\to\infty} \frac 1k E\left(\log G(W_{N_k})\right)
	&=\frac{d}{d\alpha}\LambdaN(\LambdaG(\alpha))|_{\alpha=0}\\
 	&= \muN \HG,
\end{align*}
where 
\begin{align*}
\muN=\lim_{k\to\infty} \frac{E(N_k)}{k},
\text{ } \HG=-\sum_{i\in\A} P(W_1=i)\log P(W_1=i). 
\end{align*}
The growth rate of the average subordinated guesswork is, however,
given by the sCGF of the noise evaluated at the character R\'enyi entropy 
at $1/2$,
\begin{align*}
\lim_{k\to\infty} \frac 1k \log E\left(G(W_{N_k})\right)
	&=\LambdaN(\LambdaG(1))\\
 	&= \lim_{k\to\infty} \frac 1k \log E\left(\exp\left(R(1/2) N_k\right)\right).
\end{align*}
\end{corollary}
Thus the determining factor in the average guesswork is not the
average noise, but the scaling of the cumulant of the noise process
determined by the R\'enyi entropy with parameter $1/2$.

\section{Examples}

Corollary \ref{cor:growth} suggests the design of schemes whose
principle is to ensure that the stochastic properties of either the
the noise or the source, as manifested through the behavior of
$\LambdaN(\LambdaG(1))$, differs between intended receivers
and unintended receivers so as to provide a lower growth rate
in the average subordinated guesswork of intended receivers. The
intended and unintended receivers may observe the transmitted string
through parallel channels or through a common channel where there
exists a dependence in the noise at different receivers. Such
scenarios mirror those that have been considered for secrecy capacity,
with the latter having been extensively studied as a model for
wireless channels in which the unintended receivers are eavesdroppers
(e.g. \cite{Bloch11}) and the former considered less often
\cite{Yamamoto86, Yamamoto91, Zang06}.

In order to reduce the guesswork of intended receivers over
eavesdroppers, common randomness between the word subordinated by
the erasures process and the intended receivers may be used as
common randomness is a means of generating a secret key \cite{Ahlswede93,
Maurer93, Csiszar00}. Common randomness for secrecy may be derived
from the source itself, for instance as side information regarding
the source. Channel side information is commonly proposed or used
as a source of common randomness in wireless networks (e.g.
\cite{Mathur10} and references therein). The use of differentiated
channel side information for guesswork in erasure channels provides
a means of tuning properties of the sCGF of the noise, $\LambdaN$.

In wireless erasure channels, we may consider several means of
achieving differentiated channel side information between intended
receivers and eavesdroppers. Consider, for example, a fading channel,
where fades lead to erasures and where fading characteristics permit
prediction of future fades from current channel measurements. A
node that actively sounds the channel, or receives channel side
information from a helper node, may know, perfectly or imperfectly,
which erasures will occur over some future time. 

Friendly jamming instantiates different channel side information
between intended and unintended receivers by actively modifying the
channel. Friendly jamming have been proposed and demonstrated to
modify secrecy regions in wiretap-like settings \cite{Vilela11,
Gollakota11}. A notion related to friendly jamming is that of
cooperative jamming \cite{Tekin08} where multiple users collude in
their use of the shared channel in order to reduce an eavesdropper's
ability. In our setting, a jammer using a pseudo-noise sequence
that is known to the receiver but appearing Bernoulli to the
eavesdropper can render the average guesswork of the intended
receiver to be lower than that of the eavesdropper. Such friendly
jamming may be effective even if the jammer generates, on average,
more erasures for the intended receiver than for the eavesdropper,
for instance because it may be closer to the former than to the
latter. The same mechanism can be used to create attacks in which
a jammer, using a sequence known to the eavesdropper but not to the
receiver, may increase the guesswork of the intended receiver
relatively to that of the receiver.

We explore numerically the effect of modifying the distribution of
the source or erasures, for instance through the use of friendly
jamming. Corollary \ref{cor:growth} makes clear that the growth
rate of the average subordinated guesswork depends upon unusual
statistics of both the noise and of the character source. We
illustrate the ramifications of this by demonstrating that an
unintended eavesdropper can have a better channel, on average, yet
have a larger average guesswork than an intended recipient with a
noisier average channel. We illustrate this in two ways: with the
character distribution fixed and varying the noise, and vice versa.

Consider the simplest example, where an unintended recipient has
an i.i.d. channel with a probability of erasure of $p$ per character.
This is a typical channel model and gives
\begin{align*}
\LambdaN(\beta)=\log(1-p+pe^\beta).
\end{align*}
Thus his average subordinated guesswork growth rate is
\begin{align*}
\LambdaN(R(1/2))=\log(1-p+pe^{R(1/2)}) 
\geq p R(1/2),
\end{align*}
where the latter follows by Jensen's inequality and unless $p=0$
or $1$ the inequality is strict. 

If the intended receiver has a deterministic channel with a proportion
$\mu$ of characters erased, then the growth rate of its average
subordinated guesswork is $\mu R(1/2)$. In particular, if
$p<\mu<R(1/2)^{-1}\log(1-p+p\exp(R(1/2)))$ then even though the
channel of the unintended recipient is, on average, more noisy than
the intended recipient, the average guesswork of the latter is
smaller.

\begin{figure}
\vspace{-2cm}
\includegraphics[scale=0.35]{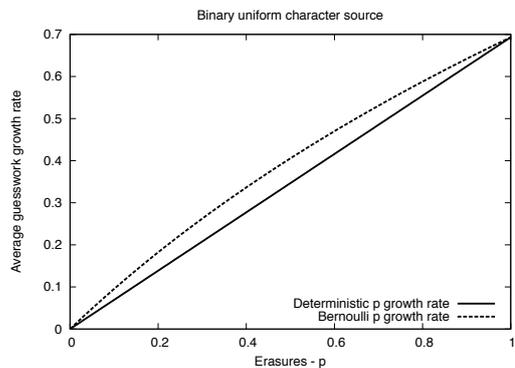}
\vspace{-1.5cm}
\caption{Binary source alphabet, $\A=\{1,2\}$, with $P(W_1=1)=1/2$. 
Average guesswork growth rate for deterministic channel with
proportion $p$ characters erased compared to a memoryless Bernoulli
$p$ erasure channel. For a given average number of erasures, the
deterministic channel has a lower average guesswork.
}
\label{fig:fixed_source}
\end{figure}

Assuming a binary alphabet, $\A=\{1,2\}$, we present three figures
to illustrate that the average guesswork depends upon both the
channel and source statistics. First, fix the source statistics by
assuming $P(W_1=1)=1/2$. For $p\in[0,1]$, Figure \ref{fig:fixed_source}
plots the average guesswork growth rate for the deterministic channel
$p R(1/2)$ and for the Bernoulli channel $\log(1-p+p\exp(R(1/2)))$.
If $p\neq0$ or $1$, the Bernoulli channel has a higher average
guesswork. Thus the intended recipient could have, on average, a
less noisy channel, yet have a lower average guesswork. For clarity,
Figure \ref{fig:fixed_source_gap} plots the difference between these
growth rates.

\begin{figure}
\vspace{-2cm}
\includegraphics[scale=0.35]{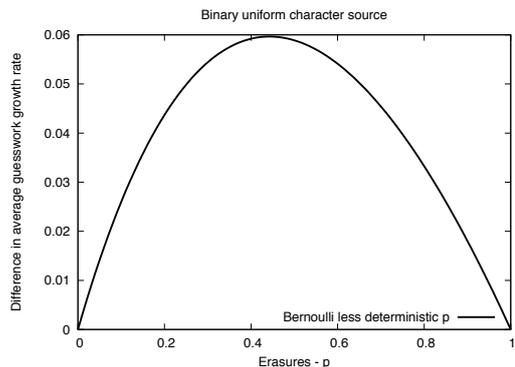}
\vspace{-1.5cm}
\caption{Binary source alphabet, $\A=\{1,2\}$, with $P(W_1=1)=1/2$. 
Similar to Figure \ref{fig:fixed_source}, but plotting the difference
between the Bernoulli $p$ average guesswork growth rate and the
deterministic $p$ average guesswork. 
}
\label{fig:fixed_source_gap}
\end{figure}

Figures \ref{fig:fixed_source} and \ref{fig:fixed_source_gap}
highlight the influence of the channel statistics on the average
guesswork growth rate, but Figure \ref{fig:fixed_noise} demonstrates
the confounding influence of the source statistics. Here we assume
that one channel is deterministic with $14\%$ of characters erased
while the other channel is Bernoulli with an average of $10\%$
characters erased. Figure \ref{fig:fixed_noise} plots the difference
in average guesswork growth rate between these two channels as the
source statistics change. If the source is less variable, the
deterministic channel has a higher average guesswork, but as the
source statistics become more variable, this reverses and the Bernoulli
channel has higher average guesswork growth rate. In other words,
even though the average noise on the deterministic channel is worse,
dependent upon the source statistics its average guesswork may
be lower than a Bernoulli channel with lower average noise.

\begin{figure}
\vspace{-2cm}
\includegraphics[scale=0.35]{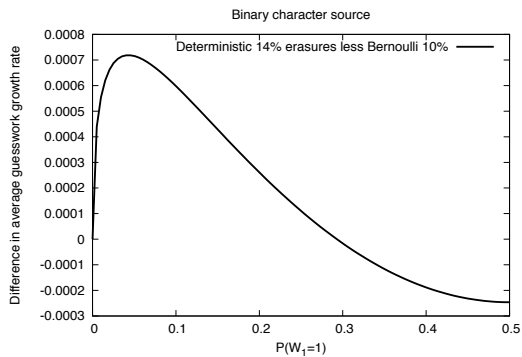}
\vspace{-1.5cm}
\caption{Binary source alphabet, $\A=\{1,2\}$. Difference in average
guesswork growth rate, as a function of $P(W_1=1)$, between a
deterministic channel with $14\%$ characters erased and a Bernoulli
channel with $10\%$ chance that each character is deleted. If the
character source is less variable, the deterministic channel has a
higher growth rate, but as the character source becomes more variable,
it has a lower growth rate.
}
\label{fig:fixed_noise}
\end{figure}

Between them, these examples indicate the trade-off in influence
of the source and noise statistics on the guesswork. While we have
assumed the simplest noise channels, these results are characteristic
of the system. 

\section{Conclusions}

We have characterized the asymptotic distribution of the guesswork
required to reconstitute a string that has been subject to symbol
erasure, as occurs on noisy communication channels.
The scaled Cumulant Generating Function of the guesswork subordinated
by the erasure process has a simple characterization as the composition
of the sCGF of the noise with the sCGF of the unsubordinated guesswork.
This form is redolent of the well-known result for the moment
generating function for a random sun of random summands, but is an
asymptotic result for guesswork. These results suggest that methods
inspired from the secrecy capacity literature, such as the use of
differentiated channel or source side information between the
intended receiver and the eavesdropper, can be used in the context
of guesswork. Indeed, numerical examples show that deterministic
erasures can lead to lower average guesswork than Bernoulli erasures
with a lower mean number of erasures. In further work, one may
consider the behavior of guesswork in different settings that have
been explored in the wiretap and cognate literature.

One may also envisage generalizing this analysis to the case where
there are retransmissions of the entire string or of the symbols
that have not been received by the intended receiver. Retransmissions
are commonly employed in several protocols to enable reliability
and, in the case of an erasure channel with perfect feedback, taking
the form of acknowledgments, uncoded retransmission is capacity-achieving.


\end{document}